\documentclass{article}

\usepackage{booktabs} 
\usepackage{graphicx}
\usepackage{float}
\usepackage{wrapfig}
\usepackage{hyperref}
\usepackage[algo2e,linesnumbered]{algorithm2e}
\usepackage{verbatim}
\usepackage{fullpage}
\usepackage{amsmath}
\usepackage{amsthm}
\usepackage{tikz}
\usepackage{amsfonts}
\usepackage{amssymb}

\newtheorem{theorem}{Theorem}

\newtheorem{Lemma}{Lemma}
\newtheorem{Corollary}{Corollary}
\newtheorem{Definition}{Definition}
\newtheorem{Claim}{Claim}

\providecommand{\Pr}{\text{Pr}}
\renewcommand{\Pr}{\text{Pr}}

\renewcommand{\r}{\mbox{rank}}

\newcommand{\facloc}{\textsc{FacLoc}}
\newcommand{\pmedian}{$p$\textsc{Median}}
\newcommand{\pcenter}{$p$\textsc{Center}}
\renewcommand{\to}{\tilde{O}}

\providecommand{\rdown}{r}
\renewcommand{\rdown}{r}
\providecommand{\dup}{d}
\renewcommand{\dup}{d}
\providecommand{\rsim}{r}
\renewcommand{\rsim}{r}

\newcommand*\samethanks[1][\value{footnote}]{\footnotemark[#1]}

\begin{document}
\title{Near-Optimal Clustering in the $k$-machine model}

\author{Sayan Bandyapadhyay \thanks{Department of Computer Science, The University of Iowa \texttt{\{sayan-bandyapadhyay, tanmay-inamdar, shreyas-pai, sriram-pemmaraju\}@uiowa.edu}} \and Tanmay Inamdar \samethanks[1] \and Shreyas Pai \samethanks[1] \and Sriram V.~Pemmaraju \samethanks[1]}



\maketitle

\begin{abstract}
	The \textit{clustering problem}, in its many variants, has numerous applications in operations research and computer science (e.g., in applications in bioinformatics, image processing, social network analysis, etc.).
	As sizes of data sets have grown rapidly, researchers have focused on designing algorithms for clustering problems in 
	models of computation suited for large-scale computation such as MapReduce, Pregel, and streaming models.
	The \textit{$k$-machine model} (Klauck et al., SODA 2015) is a simple, message-passing model for large-scale distributed graph processing.
	This paper considers three of the most prominent examples of clustering problems: the \textit{uncapacitated facility location}
	problem, the \textit{$p$-median} problem, and the \textit{$p$-center} problem and presents $O(1)$-factor approximation
	algorithms for these problems running in $\tilde{O}(n/k)$ rounds in the $k$-machine model. 
	These algorithms are optimal up to polylogarithmic factors because this paper also shows $\tilde{\Omega}(n/k)$ lower bounds for
	obtaining polynomial-factor approximation algorithms for these problems.
	These are the first results for clustering problems in the $k$-machine model.

	We assume that the metric provided as input for these clustering problems in only 
	implicitly provided, as an edge-weighted graph and in a nutshell, our main technical contribution is to show that constant-factor approximation algorithms for all three clustering
	problems can be obtained by learning only a small portion of the input metric.
\end{abstract}




\IncMargin{0.5em}
\section{Introduction}

The problem of \textit{clustering} data has a wide variety of applications in areas such as information retrieval,
bioinformatics, image processing, and social network analysis. In general, clustering is a key component of data mining 
and machine learning algorithms. Informally speaking, the objective of data clustering is to partition data into groups
such that data within each group are ``close'' to each other according to some similarity measure. 
For example, we might want to partition visitors to an online retail store (e.g., Amazon) into groups of
customers who have expressed preferences for similar products.
As the sizes of data sets have grown significantly over the last few years, it has become imperative that clustering
problems be solved efficiently in models of computation that allow multiple machines to process data in parallel.
Distributing input data across multiple machines is important not just for speeding up computation through parallelism, but also because
no single machine may have sufficiently large memory to hold a full data set. 
Motivated by these concerns, recent research has considered problems of designing clustering algorithms 
\cite{EneIMKDD2011}\cite{GarimellaDGSCIKM2015} in systems such as MapReduce \cite{DeanGhemawatCACM2010} and Pregel \cite{MalewiczABDHLCSIGMOD2010}. 
Clustering algorithms \cite{SilvaFBHCCompSurv2013} have also been designed for streaming models of computation 
\cite{AlonMSSTOC1996}. 

In this paper we present distributed algorithms for three of the most prominent clustering problems: 
the \textit{uncapacitated metric facility location} problem,
the \textit{$p$-median} problem, and the \textit{$p$-center} problem.
All three problems have been studied for several decades now and are well-known to be NP-hard.
On the positive side, all three problems have constant-factor (polynomial-time) approximation algorithms.
We consider these problems in the recently proposed \textit{$k$-machine model} \cite{KlauckNPRSODA15}, a synchronous, message-passing
model for large-scale distributed computation. 
This model cleanly abstracts essential features of systems such as Pregel \cite{MalewiczABDHLCSIGMOD2010} and Giraph (see \verb+http://giraph.apache.org/+) that have been designed for large-scale graph processing\footnote{Researchers at Facebook recently used Apache Giraph to process graphs with trillion edges \cite{ChingEKLMVLDB2015}.}, allowing researchers to prove
precise upper and lower bounds.
One of the main features of the $k$-machine model is that the input, consisting of $n$
items, is randomly partitioned across $k$ machines. Of particular interest are settings in which $n$ is 
much larger than $k$. Communication occurs via bandwidth-restricted
communication links between every pair of machines and thus the underlying communication network is a size-$k$ clique.
For all three problems, we present constant-factor approximation algorithms that run in $\tilde{O}(n/k)$ rounds in 
the $k$-machine model. We also show that these algorithms have optimal round complexity, to within polylogarithmic factors,
by providing complementary $\tilde{\Omega}(n/k)$ lower bounds for polynomial-factor approximation algorithms\footnote{Throughout the paper, we use $\to(f(n))$ as a shorthand for $O(f(n) \cdot \mbox{poly}(\log n))$ and $\tilde{\Omega}(f(n))$ as a shorthand for $\Omega(f(n)/\mbox{poly}(\log n))$.}.
These are the first results on clustering problems in the $k$-machine model.

\subsection{Problem Definitions}
The input to the \textit{uncapacitated metric facility location} problem (in short, \facloc) is a set $V$ of points, a metric
$d: V \rightarrow \mathbb{R}^{+}$ that assigns distances to point-pairs, and a \textit{facility opening cost} $f: V \rightarrow \mathbb{R}^+$ associated
with each point $v \in V$. 
The problem is to find a subset $F \subseteq V$ of points to open (as ``facilities'') so as to minimize
the objective function  $\sum_{i \in F} f_i + \sum_{j \in V} d(j, F)$, where
$d(j, F) = \min_{x \in F} d(j, x)$.
(For convenience, we abuse notation and use $f_i$ instead of $f(i)$.)
\facloc\ is NP-hard and is in fact hard to approximate with an approximation factor better than 1.463 \cite{GuhaKhullerSODA1998}.
There are several well-known constant-factor approximation algorithms for \facloc\ including the primal-dual
algorithm of Jain and Vazirani \cite{JainVaziraniJACM2001} and the greedy algorithm of Mettu and Plaxton \cite{MettuPlaxtonSICOMP2003}.
The best approximation factor currently achieved by an algorithm for \facloc\ is 1.488 \cite{LiICALP2011}.

The input to the \textit{$p$-median} problem (in short, \pmedian) is a set $V$ of points and a metric
$d: V \rightarrow \mathbb{R}^{+}$ that assigns distances to point-pairs, and a positive integer $p$.
The problem is to find a subset $F \subseteq V$ of exactly $p$ points to open (as ``facilities'') so as to minimize
the objective function  $\sum_{j \in V} d(j, F)$.
\pmedian\ is NP-hard and and is in fact hard to approximate with an approximation factor better than $1 + \frac{2}{e}
\approx 1.736$ \cite{JainMSSTOC2002}.
A well-known approximation algorithm for the $p$-median problem is due to Jain and Vazirani \cite{JainVaziraniJACM2001}, who present
a 6-approximation algorithm. This approximation factor has been improved by subsequent results -- see 
\cite{AryaGKMMPSTOC2001}, for example.
The input to the \textit{$p$-center} problem (in short, \pcenter) is the same as the input to \pmedian,
but the objective function that is minimized is $\max_{j \in V} d(j, F)$.
Like \facloc\ and \pmedian, the \pcenter\ problem is not only NP-hard, it is in fact hard to approximate 
with an approximation factor strictly better than 2 \cite{GonzalezTCS1985}.
There is also an optimal 2-approximation algorithm for this problem \cite{GonzalezTCS1985} obtained via 
a simple, greedy technique called \textit{farthest first traversal}.

In all three problems, it is assumed that each point is ``connected'' to the nearest open facility. 
So an open facility along with the ``clients'' that are connected to it forms a cluster.

\subsection{The $k$-machine Model and Input-Output Specification}
Let $n$ denote $|V|$. 
The $k$-machine model is a message-passing, synchronous model of distributed computation.
Time proceeds in \textit{rounds} and in each round, each of the $k$ machine performs
local computation and then sends, possibly distinct, messages to the remaining $k-1$ machines.
A fundamental constraint of the $k$-machine model is that each message is required to be small;
as is standard, we assume here that each message is of size $O(\log n)$ bits.
It is assumed that the $k$ machines have unique \texttt{ID}s, that are represented by 
$O(\log n)$-bit strings.

As per the random partition assumption of $k$-machine model \cite{KlauckNPRSODA15},
the points in $V$ are distributed uniformly at random across the $k$ machines.
This results in $\to(n/k)$ points per machine, with high probability (w.h.p.)\footnote{We use ``with high
probability'' to refer to probability that is at least $1 - 1/n^c$ for any constant $c \ge 1$.}.
We use $m_j$, $1 \le j \le k$, to denote the machines and $H(m_j)$ to denote the subset of points ``hosted''
by $m_j$.
The natural way to distribute the rest of the input, namely $d: V \times V \rightarrow \mathbb{R}^{+}$ and $f: V 
\rightarrow \mathbb{R}^{+}$ (in the case of \facloc), is for each machine $m_j$ to be given $f_i$ and $\{d(i, x)\}_{x \in V}$ for each
point $i \in H(m_j)$.
The distribution of $f$ in this manner is fine, but there is a problem with distributing $\{d(i, x)\}_{x\in V}$
in this manner.
Since $n$ is extremely large, it is infeasible for $m_j$ to hold the $\tilde{\Omega}(n^2/k)$ elements in 
$\cup_{i \in H(m_j)} \{d(i, x)\}_{x \in V}$. (Recall that $n >> k$.)
In general, this explicit knowledge of the metric space consumes too much memory, even when divided among
$k$ machines, to be feasible.
So we make, what we call the \textit{graph-metric assumption}, that the metric $d: V \times V \rightarrow \mathbb{R}^+$
is specified implicitly by an edge-weighted graph with vertex set $V$.
Let $G = (V, E)$ be the edge-weighted graph with non-negative edge weights representing the metric $d: V \times 
V \rightarrow \mathbb{R}^+$. Thus for any $i, j \in V$, $d(i, j)$ is the shortest path distance between points $i$ 
and $j$ in $G$.

Klauck et al.~\cite{KlauckNPRSODA15} consider a number of graph problems in the $k$-machine model
and we follow their lead in determining the initial distribution of $G$ across machines.
For each point $i \in H(m_j)$, machine $m_j$ knows all the edges in $G$ incident on $i$ and 
for each such edge
$(i, x)$, machine $m_j$ knows the \texttt{ID} of the machine that hosts $x$. 
Thus, $\sum_{i \in H(m_j)} \mbox{degree}_G(i)$ elements are needed at each machine $m_j$ to represent
the metric space and if $G$ is a sparse graph, this representation can be quite compact.

The graph-metric assumption fundamentally affects the algorithms we design. Since the metric $d$ is provided 
implicitly, via $G$, access to the metric is provided through shortest path computations on $G$.
In fact, it turns out that these shortest path computations are the costliest part of our algorithms.
One way to view our main technical contribution is this: we show that for all three clustering problems,
there are constant-factor approximation algorithms that only require a small (i.e., polylogarithmic) number of calls to a subroutine that
solves the \textit{Single Source Shortest Path (SSSP)} problem.

For all three problems, the output consists of $F$, the set of open facilities, and
connections between clients (i.e., points that have not been open as facilities) and their 
nearest open facilities. More precisely, for any machine $m_j$ and any point $i \in H(m_j)$:
\begin{itemize}
\item If $i \in F$, then $m_j$ knows that $i$ has been opened as a facility and furthermore
$m_j$ knows all $(x, \texttt{ID}_x)$-pairs where $x$ is a client that connects to $i$ and $\texttt{ID}_x$ is the \texttt{ID} of the machine hosting $x$.
\item If $i \in V \setminus F$, then $m_j$ knows that $i$ is a client and it also knows 
the $(x, \texttt{ID}_x)$ pair, where $x$ is the open facility that $i$ connects to and
$\texttt{ID}_x$ is the \texttt{ID} of the machine hosting $x$.
\end{itemize}

\subsection{Our Results}
We first prove $\tilde{\Omega}(n/k)$ lower bounds (in Section \ref{section:lowerBounds}) for \facloc, \pmedian, and \pcenter.
For each problem, we show that obtaining an $\alpha$-approximation algorithm in the $k$-machine model, for any $\alpha = (\mbox{poly}(n))$, 
requires at least $\tilde{\Omega}(n/k)$ rounds.
In the subsequent three sections, we present $\to(n/k)$-round, constant-factor approximation algorithms for the \facloc, \pmedian,
and \pcenter\ problem, respectively.
Our lower bound results show that our algorithms have optimal round complexity,
at least up to polylogarithmic factors.

We bring to bear a wide variety of old and new techniques to derive our upper bound results including the facility location
algorithm of Mettu and Plaxton \cite{MettuPlaxtonSICOMP2003}, the fast version of this algorithm due to Thorup \cite{Thorup2001}, 
the neighborhood-size estimation framework of 
Cohen \cite{Cohen1997,Cohen2015}, the $p$-median Lagrangian relaxation algorithm of Jain and Vazirani \cite{JainVaziraniJACM2001} and the recent distributed shortest path algorithms 
due to Becker et al.~\cite{BeckerKKLdisc17}.
In our view, an important contribution of this paper is to show how all of these techniques can be utilized in the $k$-machine model.

\subsection{Related Work}
\label{section:relatedWork}

Following Klauck et al.~\cite{KlauckNPRSODA15}, two other papers \cite{PanduranganRSarxiv16,PanduranganRSSPAA2016} have studied
graph problems in the $k$-machine model. In \cite{PanduranganRSSPAA2016}, the authors present an $\tilde{O}(n/k^2)$-round algorithm
for graph connectivity, which then serves as the basis for $\tilde{O}(n/k^2)$-round algorithms for other graph problems such as
minimum spanning tree (MST) and approximate min-cut. The upper bound for MST does not contradict the $\Omega(n/k)$ lower
bounds shown for this problem in Klauck et al.~\cite{KlauckNPRSODA15} because Pandurangan et al.~\cite{PanduranganRSSPAA2016} use
a more relaxed notion of how the output MST is represented. Specifically, at the end of the algorithm in \cite{PanduranganRSSPAA2016}
every MST edge is known to \textit{some} machine, whereas Klauck et al.~\cite{KlauckNPRSODA15} use the stricter requirement
that every MST edge be known to the machines hosting the two end points of the edge.
This phenomena in which the round complexity of the problem is quite sensitive to the output representation may be
relevant to our resuts as well and is further discussed in Section \ref{section:conclusions}. 

Earlier in this section, we have mentioned models and systems for large-scale parallel computation such as MapReduce and Pregel.
Another model of large-scale parallel computation, that seems essentially equivalent to the $k$-machine model is
the \textit{Massively Parallel Computation model (MPC)} which according to \cite{YaroslavtsevVarxiv2017} is the
``most commonly used theoretical model of computation on synchronous large-scale
data processing platforms such as MapReduce and Spark.''

\section{Lower Bound Results}
\label{section:lowerBounds}

In this section, we derive $\tilde{\Omega}(n/k)$ lower bounds for achieving \(\mbox{poly(n)}\)-factor approximation algorithms in the $k$-machine model for all three problems considered in this paper. Our lower bounds are inspired by the \(\Omega(n/k)\) lower bound result from \cite{KlauckNPRSODA15} for the \textit{Spanning Tree Computation} problem.

To prove the lower bounds we describe a family of lower bound graphs \(F_b(X, Y)\) where \(X\) and \(Y\) are sampled from the same distribution as the one used in \cite{KlauckNPRSODA15}. That is, \((X, Y)\) is chosen uniformly at random from $\{0, 1\}^b \times \{0, 1\}^b$, satisfying the constraint that for every $i \in [b]$, $X_i + Y_i \ge 1$. Let \(b = n/2 - 1\) and let \(L = n^c\) for some large enough constant \(c\) that depends on the approximation factor considered. The graph \(F_b(X, Y)\) has \(2b + 2\) vertices \(u, w, u_1, \dots, u_b, w_1, \dots, w_b\). We fix the ID's of the vertices to be the first \(n\) natural numbers which means that each machine knows whether a vertex \(v\) is \(u, w, u_i, w_i\) just by knowing ID(v). For every \(i \in [b]\) there are three edges in the graph  of the form \(\{u, u_i\}, \{u_i, w_i\}, \{w_i, w\}\) and the weights of these edges depend on the bit values of \(X_i\) and \(Y_i\) where \(X, Y \in \{0, 1\}^b\). In particular, we assign weights to \((\{u, u_i\}, \{u_i, w_i\}, \{w_i, w\})\) as follows -- if \(X_i = 1\) and \(Y_i = 0\), the weights are \((1, 1, L)\), if \(X_i = 0\) and \(Y_i = 1\), the weights are \((L, 1, 1)\), and if \(X_i = 1\) and \(Y_i = 1\), the weights are \((1, L, 1)\). There is no weight assignment for the case when \(X_i = Y_i = 0\) because the distribution of \((X, Y)\) places no probability mass on this case.

In the following lemma we show that any protocol that reveals \(X\) and \(Y\) to a single machine must do so by making it receive large messages from other machines. The proof is the same as the entropy argument made in theorem 2.1 in \cite{KlauckNPRSODA15} with the added simplification that the entropy at the end of the protocol is zero. Nevertheless, we prove the lemma for completeness.

\begin{Lemma} \label{lem:entropymessage}
  Let \(\Pi\) be a public-coin $\epsilon$-error randomized protocol in the $k$-machine model \((k \ge 4)\) on an \(n\)-vertex input graph sampled uniformly at random from \(F_b(X, Y)\). If a machine knows both \(X\) and \(Y\) at the end of the protocol \(\Pi\) then it must receive \(\Omega(b)\) bit messages in expectation from other machines.
\end{Lemma}
\begin{proof}
  Let \(p\) be the machine that knows both \(X\) and \(Y\) at the end of the protocol. Since \(X\) and \(Y\) are encoded in the edge weights of the graph, if the machine \(p\) hosts \(u\) then it knows the string \(X\) via the edges \(\{u, u_i\}\) and similarly it knows \(Y\) if it hosts \(w\). But if \(p\) hosts both \(u\) and \(w\) then it knows \(X\) and \(Y\) before the protocol even begins. This is a bad event so we condition on the event that no machine hosts both \(u\) and \(w\) which happens with probability \(1 - 1/k\).

  Before the first round of communication, it can be shown that the entropy \(H(X, Y) \ge H(Y \mid X) = H(X \mid Y) =  2b/3\). The machine \(p\) also hosts some vertices \(u_i\) and \(w_i\) giving it access to some bits of \(X\) and \(Y\). It is easy to see via the Chernoff bound that with very high probability \(p\) hosts at most \((1 + \zeta)2b/k\) \(u_i\)'s and \(w_i\)'s for \(\zeta = 0.01\) which means it cannot know more than \((1 + \zeta)2b/k\) bits of \(X\) and \(Y\) by virtue of hosting these vertices whp. The event where \(p\) hosts more vertices cannot influence \(H(X, Y)\) the entropy by more than \(2^{-\zeta^2 2b/(3k)} \cdot b = o(1)\) for \(b\) large enough. Hence, the entropy of \(X, Y\) given this initial information (which we denote by a random variable \(A\)) is \(H(X, Y \mid A) \ge 2b/3 - (1 + \zeta)2b/k - o(1)\). Note that if \(p\) hosts either \(u\) or \(w\) then \(A\) will contain information about either \(X\) or \(Y\) respectively but that does not affect our lower bound on the initial entropy.

  Let \(\Pi_p\) be the messages received by the machine \(p\) during the course of the protocol \(\Pi\). With probability \(1 - \epsilon\), \(p\) knows both \(X\) and \(Y\) at the end of the protocol and therefore \(H(X, Y \mid \Pi_p, A) = 0\). This means that \(I(X, Y; \Pi_p | A) = H(X, Y | A) \ge  2b/3 - (1 + \zeta)b/k - o(1)\) and that \(|\Pi_p| = \Omega(b)\). This is under the assumption that different machines host \(u\) and \(w\) and there is no error, therefore the expected number of messages received by \(p\) must be at least \((1 - \epsilon) \cdot (1 - 1/k) \cdot \Omega(b) = \Omega(b)\).
\end{proof}

\begin{Lemma}\label{lem:lbfl}
For any $1 \le \alpha \le \mbox{poly}(n)$, every public-coin $\epsilon$-error randomized protocol in the $k$-machine model that computes an $\alpha$-factor approximate solution of \facloc\ on an $n$-vertex input graph has an expected round complexity $\tilde{\Omega}(n/k)$.
\end{Lemma}
\begin{proof}
  To prove the lemma we consider the family of lower bound graphs \(F_b(X, Y)\) with the additional property that the vertices \(u\) and \(w\) have facility opening cost \(0\) and every other vertex has opening cost \(L\). 

  Consider the solution \(\mathcal{S}\) to Facility Location where we open the vertices \(u\) and \(w\) and connect all other vertices to the closest open facility.
  The cost of this solution is \(O(n)\) whereas any other solution will incur a cost of at least \(\Omega(L)\). By our choice of \(L\), the solution \(\mathcal{S}\) is optimal and any \(\alpha\)-approximate solution is forced to have the same form as \(\mathcal{S}\).

  After the facility location algorithm terminates, with probability \(1 - \epsilon\), the machine \(p\) hosting \(u\) will know the ID's of the \(w_i\)'s that \(u\) serves in \(\mathcal{S}\). This allows \(u\) to figure out \(Y\) because \(Y_i = 0\) if \(u\) serves \(w_i\) and \(Y_i = 1\) otherwise. By Lemma \ref{lem:entropymessage}, \(p\) receives \(\Omega(b)\) bit messages in expectation throughout the course of the algorithm. This implies an \(\tilde{\Omega}(n/k)\) lower bound on the expected round complexity.
\end{proof}

\begin{Lemma}
For any $1 \le \alpha \le \mbox{poly}(n)$, every public-coin $\epsilon$ error randomized protocol on a $k$-machine network that computes a $\alpha$-factor approximate solution of $\pmedian$ and $\pcenter$ on an $n$-vertices input graph has an expected round complexity of $\tilde{\Omega}(n/k)$.
\end{Lemma}

\begin{proof}
  We show the lower bound for \(p=2\) on graphs that come from the family \(F_b(X, Y)\). An optimal solution in a graph from this family is to open \(u\) and \(w\) which gives a solution of cost \(O(n)\) for \(\pmedian\) and \(O(1)\) for \(\pcenter\). But, we need to be a bit more careful because the \(\pmedian\) or \(\pcenter\) algorithms can choose to open some of the \(u_i\)'s and \(w_j\)'s instead of \(u\) and \(w\) with only a constant factor increase in the cost of the solution. More specifically, there are four possible cases where we can open different pairs of vertices to get an \(O(1)\)-approximate solution -- \((u, w)\), \((u_i, w)\), \((u_i, w_j)\), and \((u, w_j)\) where \(u_i\) and \(w_j\) are connected by an edge of weight \(1\) to \(u\) and \(w\) respectively. In all these cases, the opened vertices know both \(X\) and \(Y\) at the end of the algorithm by virtue of knowing the vertices that it serves in the final solution. This is because the value of \(L\) is high enough to ensure that the two clusters formed in any \(\alpha\)-approximate solution are the same as the optimal solution no matter what centers are chosen. Therefore, we can apply lemma \ref{lem:entropymessage} to all these cases which gives us that the machine hosting one of these vertices will receive \(\Omega(b)\) bit messages in expectation during the course of the algorithm. This means that the expected round complexity for both the algorithms is \(\tilde{\Omega}(n/k)\).
\end{proof}

 \section{Technical Preliminaries}
\label{section:preliminaries}

Since the input metric is only implicitly provided, as an edge-weighted graph, computing shortest path distances to learn
parts of the metric space turns out to be a key element of our algorithms.
The \textit{Single Source Shortest Path (SSSP)} problem has been considered in the $k$-machine model in Klauck et al.~\cite{KlauckNPRSODA15}
and they describe a $(1+\epsilon)$-approximation algorithm that runs in the $k$-machine model in 
$\to(n/\sqrt{k})$ rounds.
This is too slow for our purpose, since we are looking for an overall running time of $\to(n/k)$.
We instead turn to a recent result of Becker at al.~\cite{BeckerKKLdisc17} and using this we can easily obtain an $\to(n/k)$-round SSSP
algorithm.
Becker et al.~do not work in the $k$-machine model; their result relevant to us is in the 
\textit{Broadcast Congested Clique} model.
Informally speaking, the \textit{Congested Clique} model can be thought of as a 
special case of the $k$-machine model with $k = n$.
The \textit{Broadcast Congested Clique} model imposes the additional restriction on communication
that in each round each machine sends the \textit{same} message (i.e., broadcasts) to the 
remaining $n-1$ machines.
We now provide a brief description of the Congested Clique models.
The \textit{Congested Clique} model consists of $n$ nodes (i.e., computational entities)
connected by a clique communication network. Communication is point-to-point via
message passing and each message can be at most $O(\log n)$ bits in length.
Computation proceeds in synchronous rounds and in each round, each node performs
local computations and sends a (possibly different) message to each of the other
$n-1$ nodes in the network.
For graph problems, the input is assumed to be a spanning subgraph of the underlying
clique network and each node is initially aware of the incident edges in the input.
The \textit{Broadcast Congested Clique} model differs from the Congested Clique model
only in that in each round, each node is required to send the same message to the
remaining $n-1$ nodes.
For more details on the Congested Clique models, see \cite{HegemanPPSSPODC15,DruckerKOPODC2014}.

\begin{theorem} \textbf{(Becker et al.~\cite{BeckerKKLdisc17})}
For any $0 < \epsilon \le  1$, in the Broadcast Congested Clique model, a deterministic $(1 + \epsilon)$-approximation to 
the SSSP problem in undirected graphs with non-negative edge-weights can be computed in 
	$\mbox{poly}\,(\log n)/\mbox{poly}\,(\epsilon)$ rounds. 
\end{theorem}
It is easy to see that any Broadcast Congested Clique algorithm that runs in $T$ rounds can be simulated in the $k$-machine model in
$T \cdot \to(n/k)$ rounds. 
A more general version of this claim is proved in Klauck et al.~in the Conversion Theorem (Theorem 4.1 \cite{KlauckNPRSODA15}).
This leads to the following result about the SSSP problem in the $k$-machine model.
\begin{Corollary} 
\label{theorem:SSSP}
For any $0 < \epsilon \le  1$, there is a deterministic $(1 + \epsilon)$-approximation algorithm in the $k$-machine model for solving 
the SSSP problem in undirected graphs with non-negative edge-weights in
	$O((n/k) \cdot \mbox{poly}(\log n)/\mbox{poly}(\epsilon))$ rounds.
\end{Corollary} 

In addition to SSSP, our clustering algorithms require an efficient solution to a more general problem that we
call \textit{Multi-Source Shortest Paths} (in short, MSSP).
The input to MSSP is an edge-weighted graph $G = (V, E)$, with non-negative edge-weights, and a set $T \subseteq V$ of sources.
The output is required to be, for each vertex $v$, the distance $d(v, T)$ (i.e., $\min\{d(v, u) \mid u \in T\}$) 
and the vertex $v^* \in T$ that realizes this distance.
The following lemma uses ideas from Thorup \cite{Thorup2001} to show that MSSP can be reduced to a single call to SSSP 
and can be solved in an approximate sense in the $k$-machine model in $\to(n/k)$ rounds. 

\begin{Lemma} \label{lemma:MSSP}
Given a set \(T \subseteq V\) of sources known to the machines (i.e., each machine $m_j$ knows $T \cap H(m_j)$), we can, for any value $0 \le \epsilon \le 1$,
compute a $(1 + \epsilon)$-approximation to MSSP in $\to(1/\mbox{poly}(\epsilon) \cdot n/k)$ rounds, w.h.p.
Specifically, after the algorithm has ended, for each $v \in V \setminus T$, the machine $m_j$ that hosts $v$ knows a pair
$(u, \tilde{d}) \in T \times \mathbb{R}^+$, such that
	\(d(v, u) \le \tilde{d} \le (1+\epsilon) \cdot d(v, T)\).
\end{Lemma}
\begin{proof}
	  First, as in \cite{Thorup2001}, we add a dummy source vertex $s$, and connecting $s$ to each vertex $u \in T$ by $0$-weight edges. The shortest path distance from $s$ to any other vertex $v \in V$, is same as $d(v, T)$ in the original graph. 
This dummy source can be hosted by an arbitrary machine and the edge information can be exchanged in \(\to(n/k)\) rounds 

	     Using Theorem \ref{theorem:SSSP}, we can compute approximate shortest path distance $\tilde{d}$ that satisfies the first property of the lemma, in $\tilde{O}(n/k)$ rounds. By \cite{BeckerKKLdisc17} (Section 2.3) we can compute an approximate shortest path tree in addition to approximate distances in the Broadcast Congested Clique in $O(\mbox{poly}(\log n)/\mbox{poly}\,(\epsilon))$ rounds w.h.p.~and hence in the \(k\)-machine model in $\tilde{O}(1/\mbox{poly}(\epsilon) \cdot n/k)$ rounds w.h.p.

	Since a tree contains linear (in $n$) number of edges, all machines 
	can exchange this information in $\to(n/k)$ rounds so that every machine knows the computed approximate shortest path tree. Now, each machine $m_j$ can determine locally, for each vertex $v \in H(m_j)$ the vertex $u \in T$ which satisfies the properties stated in the lemma.
\end{proof}

Note that in the solution to MSSP, for each $v \in T$, $d(v, T) = 0$.
For our algorithms, we also need the solution to a variant of MSSP that we call \textsc{ExclusiveMSSP} in which
for each $v \in T$, we are required to output $d(v, T\setminus\{v\})$ and the vertex $u^* \in T \setminus \{v\}$ that realizes this
distance. 
The following lemma uses ideas from Thorup \cite{Thorup2001} to show that \textsc{ExclusiveMSSP} can 
be solved by making $O(\log n)$ calls to a subroutine that solves SSSP.
\begin{Lemma} \label{lemma:ExclusiveMSSP}
Given a set \(T \subseteq V\) of sources known to the machines (i.e., each machine $m_j$ knows $T \cap H(m_j)$), we can, for any value $0 \le \epsilon \le 1$,
compute a $(1 + \epsilon)$-approximation to \textsc{ExclusiveMSSP} in $\to(1/\mbox{poly}\,(\epsilon) \cdot n/k)$ rounds, w.h.p.
Specifically, after the algorithm has ended, for each $v \in T$, the machine $m_j$ that hosts $v$ knows a pair
	$(u, \tilde{d}) \in T \setminus \{v\} \times \mathbb{R}^+$, such that $d(v, u) \le \tilde{d} \le (1+\epsilon) \cdot d(v, T \setminus \{v\})$.
\end{Lemma}

\begin{proof}
  Breaking ties by machine ID, each vertex in \(T\) is assigned a \(\log |T|\) size bit vector. We create \(2\log |T|\) subsets of \(T\) by making two sets \(T_i^0\) and \(T_i^1\) for each bit position \(i\). The set \(T_i^b\) contains vertices whose \(i^{th}\) bit value is \(b\). Note that for all pairs of vertices \(v, w\), there is at least one set \(T_i^b\) such that \(v \in T_i^b\) and \(w \notin T_i^b\). Now we run an MSSP algorithm for each \(T_i^b\) using lemma \ref{lemma:MSSP}. Now for each vertex \(v \in T\) \(\tilde{d}\) is the smallest \(d(v, T_i^b)\) such that \(v \notin T_i^b\) and the vertex \(u\) is an arbitrary vertex that realizes the distance \(\tilde{d}\).

\end{proof}
 \section{Facility Location in \(\tilde{O}(n/k)\) rounds}
\label{section:facilityLocation}

At the heart of our $k$-machine algorithm for \facloc\ is the well-known sequential algorithm of Mettu and Plaxton \cite{MettuPlaxtonSICOMP2003}, that computes a 
3-approximation for \facloc.
To describe the Mettu-Plaxton algorithm (henceforth, MP algorithm), we need some notation. For each real $r \ge 0$ and vertex $v$, define the ``ball'' $B(v, r)$ as
the set $\{u \in V \mid d(v, u) \le r\}$.
For each vertex $v \in V$, we define a \textit{radius} $r_v$ as the solution
$r$ to the equation $f_v = \sum_{u \in B(v, r)} (r - d(v, u))$.  Figure \ref{fig:rv} illustrates the definition of $r_v$ (note that
$r_v$ is well-defined for every vertex $v$).
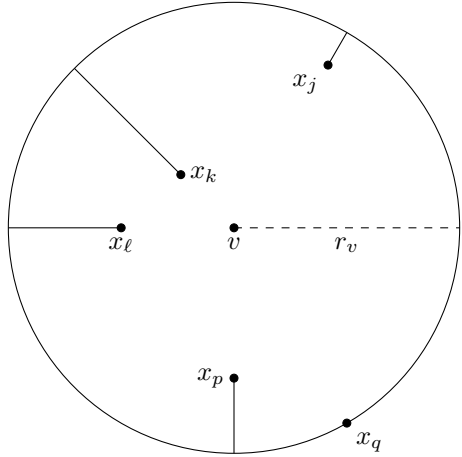
\begin{figure}
\begin{center}
\begin{tikzpicture}[scale=1,auto,swap]
    \draw (0,0) circle (3cm);
    \draw (0,0) coordinate (i);
    \draw (60:25mm) coordinate (v);
    \draw (135:1cm) coordinate (w);
    \draw (180:15mm) coordinate (x);
    \draw (270:2cm) coordinate (y);
    \draw (300:3cm) coordinate (z);
    \path[draw,dashed,-] (i) -- node[below] {$r_v$} (0:3cm);
    \draw (v) -- (60:3cm);
    \draw (w) -- (135:3cm);
    \draw (x) -- (180:3cm);
    \draw (y) -- (270:3cm);
    \node[circle,draw,fill=black!100,inner sep=0pt,minimum width=3pt] at (i) {};
    \node[below] at (i) {$v$};
    \node[circle,draw,fill=black!100,inner sep=0pt,minimum width=3pt] at (v) {};
    \node[below left] at (v) {$x_j$};
    \node[circle,draw,fill=black!100,inner sep=0pt,minimum width=3pt] at (w) {};
    \node[right] at (w) {$x_k$};
    \node[circle,draw,fill=black!100,inner sep=0pt,minimum width=3pt] at (x) {};
    \node[below] at (x) {$x_\ell$};
    \node[circle,draw,fill=black!100,inner sep=0pt,minimum width=3pt] at (y) {};
    \node[left] at (y) {$x_p$};
    \node[circle,draw,fill=black!100,inner sep=0pt,minimum width=3pt] at (z) {};
    \node[below right] at (z) {$x_q$};
\end{tikzpicture}
\end{center}
        \caption{This illustration, which originally appeared in \cite{HegemanPemmarajuDC2015}, shows
        $B(v, r_v)$, the radius-$r_v$ ball centered at
        $v$. If we imagine the ball $B(v, r)$ growing with increasing $r$ and we reach a
        stage at which $r = r_v$, then the sum
        of the 5 distances, denoted by solid line segments from points within the
        ball to the ball-boundary equals $f_v$.}
\label{fig:rv}
\end{figure}

The MP algorithm is the following simple, 2-phase, greedy algorithm:
\RestyleAlgo{boxruled}
\begin{algorithm2e}\caption{\textsc{MP} Algorithm\label{alg:MP}}
	\textbf{Radius Computation Phase.} For each vertex $v \in V$, compute $r_v$.\\
	\textbf{Greedy Phase.} Consider vertices $v \in V$ in non-decreasing order of radii $r_v$.  Starting with \(S = \emptyset\), add $v$ to $S$ if $d(v, S) > 2r_v$.
\end{algorithm2e}

We will work with a slight variant of the MP algorithm, called MP-$\beta$ in \cite{ArcherRSESA2003}.
The only difference between the MP algorithm and the MP-$\beta$ algorithm is in the definition of each radius $r_v$, which is 
defined for the MP-$\beta$ algorithm, as the value $r$ satisfying $\beta \cdot f_v = \sum_{u \in B(v, r)} (r - d(v, u))$.
(Thus, the MP-$\beta$ algorithm with $\beta = 1$ is just the MP algorithm.)

There are two challenges to implementing the MP-$\beta$ algorithm efficiently in the $k$-machine model (and more generally in a distributed or 
parallel setting): (i) The calculation of the radius $r_v$ by the machine hosting vertex $v$ requires that the machine know distances
$\{d(v, u)\}_{u \in V}$; however the distance metric is initially unknown and is too costly to fully calculate, and (ii) the Greedy Phase
seems inherently sequential because it considers vertices one-by-one in non-decreasing order of radii; implementing this algorithm as-is would be too slow.
In the next three sections, we describe how to overcome these challenges and we end the section
with a complete description of our \facloc\ algorithm in the $k$-machine model.

\subsection{Reducing Radius Computation to Neighborhood-Size Computation}
To deal with the challenge of computing radii efficiently, without full knowledge of the metric, we use Thorup's approach \cite{Thorup2001}. 
Thorup works in the sequential setting, but like us, he assumes that the distance metric is implicitly specified via an edge-weighted graph.
He shows that it is possible to implement the MP algorithm in $\to(m)$ time on an $m$-edge graph.
In other words, it is possible to implement the MP algorithm without computing the 
full distance metric (e.g., by solving the \textit{All Pairs Shortest Path (APSP)} problem).
We now show how to translate Thorup's ideas into the $k$-machine model.
(We note here that Thorup's ideas for the \facloc\ problem have already been used to design algorithms in ``Pregel-like'' distributed systems \cite{GarimellaDGSCIKM2015}.)

For some $\epsilon > 0$, we start by discretizing the range of possible radii values using non-negative integer powers of $(1 + \epsilon)$.\footnote{Without loss of generality we assume that all numbers in the input,
i.e., $\{f_v\}_{v \in V}$ and $d(u, v)_{u, v \in V}$, are all at least 1. $O(1)$ rounds of preprocessing suffices
to normalize the input to satisfy this property. This guarantees that the minimum radius $r_v \ge 1$.}
For any vertex $v$ and for any integer $i \ge 1$, let $q_i(v)$ denote $|B(v, (1+\epsilon)^i)|$, the size of the neighborhood of 
$v$ within distance $(1+\epsilon)^i$.
Further, let $\alpha(v, r)$ denote the sum $\sum_{u \in B(v, r)} (r - d(v, u))$.
Now note that if $r$ increases from $(1+\epsilon)^i$ to $(1+\epsilon)^{i+1}$, then $\alpha(v, r)$ increases by at least 
$q_i(v) \cdot ((1+\epsilon)^{i+1} - (1+\epsilon)^i)$. 
This implies that $\sum_{i = 0}^{t-1} q_i(v) \cdot ((1+\epsilon)^{i+1} - (1+\epsilon)^i)$ is a lower bound on $\alpha(v, (1+\epsilon)^t)$.
This observation suggests that we might be able to use, as an approximation to $r_v$, the smallest value $(1+\epsilon)^{t-1}$
for which this lower bound on $\alpha(v, (1+\epsilon)^t)$ exceeds $f_v$.
Denote by $\tilde{r}_v$, this approximation of $r_v$.
In other words, $\tilde{r}_v := (1+\epsilon)^{t-1}$, where $t \ge 1$ is the smallest integer such that 
$\sum_{i = 0}^{t-1} q_i(v) \cdot ((1+\epsilon)^{i+1} - (1+\epsilon)^i) > f_v$.
It is not hard to show that $\tilde{r}_v$ is a good approximation to $r_v$ in the following
sense.
\begin{Lemma}
\label{lemma:approxRadii1}
For all $v \in V$, $\frac{r_v}{1+\epsilon} \le \tilde{r}_v \le r_v(1+\epsilon)$.
\end{Lemma}
\begin{proof}
  The values \(\tilde{r}_v\) and \(r_v\) respectively depend on how \(\sum_{i = 0}^{t-1} q_i(v) \cdot ((1+\epsilon)^{i+1} - (1+\epsilon)^i)\) and \(\alpha(v, r_v) = \sum_{u \in B(v, r)} (r - d(v, u))\) relate to \(f_v\).

  Recall that $q_i(v) = |B(v, (1+\epsilon)^i)|$. Following calculations show that $\sum_{i = 0}^{t-1} q_i(v) \cdot ((1+\epsilon)^{i+1} - (1+\epsilon)^i)$ can be interpreted as $\sum_{u \in B(v, (1+\epsilon)^t)} ((1+\epsilon)^t - d^{\uparrow}(v, u))$ where $d^{\uparrow}(v, u)$ is $d(v, u)$ rounded up to nearest power of \((1 + \epsilon)\).

  \begin{align*}
	  \sum_{i = 0}^{t-1} q_i(v) \cdot ((1+\epsilon)^{i+1} - (1+\epsilon)^i) &= ((1+\epsilon)^t - 1) + \sum_{i = 1}^{t-1} \left[\left|B(v, (1+\epsilon)^i) \setminus B(v, (1+\epsilon)^{i-1})\right| \cdot \left((1+\epsilon)^{t} - (1+\epsilon)^i\right)\right] \\ 
    &= ((1+\epsilon)^t - 1) + \sum_{j=1}^{t}\sum_{u \in B(v, (1+\epsilon)^j) \setminus B(v, (1+\epsilon)^{j-1})} (1+\epsilon)^t - (1+\epsilon)^{j} \\
    &= \sum_{u \in B(v, (1+\epsilon)^t)} ((1+\epsilon)^t - d^{\uparrow}(v, u))
  \end{align*}

  Therefore, we can say that--
  \[(1 + \epsilon)\alpha(v, (1+\epsilon)^{t-1}) \le \sum_{i = 0}^{t-1} q_i(v) \cdot ((1+\epsilon)^{i+1} - (1+\epsilon)^i) \le \alpha(v, (1+\epsilon)^t)\] 

  Which implies --
  \[\alpha(v, (1+\epsilon)^{t-1}) \le \sum_{i = 0}^{t-1} q_i(v) \cdot ((1+\epsilon)^{i+1} - (1+\epsilon)^i) \le \alpha(v, (1+\epsilon)^t)\] 

  Note that by definition of $\tilde{r}_v$, if $\tilde{r}_v = (1+\epsilon)^{t-1}$ then $\sum_{i = 0}^{t-1} q_i(v) \cdot ((1+\epsilon)^{i+1} - (1+\epsilon)^i) > f_v$ and $\sum_{i = 0}^{t-2} q_i(v) \cdot ((1+\epsilon)^{i+1} - (1+\epsilon)^i) \le f_v$. Thus, there has to exist a value $r_v \in [(1 + \epsilon)^{t-2}, (1 + \epsilon)^t]$ such that $\alpha(v, r_v) = f_v$ and this is the \(r\)-value computed by the MP algorithm. Since $\tilde{r}_v = (1+\epsilon)^{t-1}$, the Lemma follows.
\end{proof}
From the definition of $\tilde{r}_v$ one can see that in order to compute these values, we only require knowledge of $q_i(v)$ for all $i \ge 0$,
rather than actual distances $d(v, u)$ for all $u \in V$.
We now state the high-level $k$-machine model algorithm (Algorithm \ref{alg:RC}) for computing $\tilde{r}_v$ values.
\RestyleAlgo{boxruled}  
\begin{algorithm2e}\caption{\textsc{RadiusComputation} Algorithm (Version 1)\label{alg:RC}}
	\textbf{Neighborhood-Size Computation.} Each machine $m_j$ computes $q_i(v)$, for all integers $i \ge 0$ and for all vertices $v \in H(m_j)$.\\
	\textbf{Local Computation.} Each machine $m_j$ computes $\tilde{r}_v$ locally, for all vertices $v \in H(m_j)$. (Recall that $\tilde{r}_v := (1+\epsilon)^{t-1}$ where $t \ge 1$ is the smallest integer for which $\sum_{i = 0}^t q_i(v) \cdot ((1+\epsilon)^{i+1} - (1+\epsilon)^i) > f_v$.)
\end{algorithm2e}

In Algorithm \ref{alg:RC}, step 2 is just local computation, so we focus on Step 1 which requires the solution to the problem of computing neighborhood sizes. More specifically,  
we define the problem \textsc{NbdSizeComputation} as follows:
given an edge-weighted graph, with non-negative edge weights, compute the size of $B(v, d)$ for each vertex $v$ and positive real $d$.
The output to the problem in the $k$-machine model is required to be a distributed data structure 
(distributed among the $k$ machines) such that
each machine $m_j$ can answer any query ``What is $|B(v, d)|$?'' for any $v \in H(m_j)$
and any positive real $d$, using local computation.
Note that a ``trivial'' way of solving \textsc{NbdSizeComputation} is to solve APSP, but as
mentioned earlier this is too costly.
In the next subsection we show how to solve a ``relaxed'' version of this problem in the 
$k$-machine model in $\to(n/k)$ rounds, making only $O(\mbox{poly}(\log n))$ calls to
a $k$-machine SSSP algorithm.

\subsection{Neighborhood-Size Estimation in the $k$-machine Model}
\label{section:nbdSizeEstimation}

To solve \textsc{NbdSizeComputation} efficiently in the $k$-machine model, 
we turn to an elegant idea due to Cohen \cite{Cohen1997,Cohen2015}.
Motivated by certain counting problems, Cohen \cite{Cohen1997} presents a ``size-estimation 
framework,'' a general randomized method in the sequential setting.
Cohen's algorithm starts by assigning to each vertex $v$ a rank \(\r(v)\) chosen uniformly from \([0, 1]\).
These ranks induce a random permutation of the vertices. 
To compute the size estimate of a neighborhood, say $B(v, d)$, for a vertex $v$ and real $d > 0$, Cohen's algorithm finds the
smallest rank of a vertex in $B(v, d)$.
It is then shown (in Section 6, \cite{Cohen1997}) that the expected value of the smallest rank in $B(v, d)$ is $1/(1 + |B(v, d)|)$.
Thus, in expectation, the reciprocal of the smallest rank in $B(v, d)$ is (almost) identical to $|B(v, d)|$.
To obtain a good estimate of $|B(v, d)|$ with high probability, Cohen simply repeats the above-described procedure
independently a bunch of times and shows the following concentration result on the 
average estimator.

\begin{theorem} \textbf{(Cohen \cite{Cohen1997})} \label{thm:cohen}
	Let $v$ be a vertex and $d > 0$ a real.
	For $1 \le i \le \ell$, let $R_i$ denote the smallest rank of a vertex in $B(v, d)$ obtained in the $i$-th repetition of
	Cohen's neighborhood-size estimation procedure. 
	Let $\hat{R}$ be the average of $R_1, R_2, \ldots, R_\ell$.
	Let $\mu = 1/(1 + |B(v, d)|)$. Then, for any $0 < \epsilon < 1$,
	$$\Pr(|\hat{R} - \mu| \ge \epsilon \mu) = \exp(-\Omega(\epsilon^2 \cdot \ell)).$$
\end{theorem}
This theorem implies that $\ell = O(\log n/\epsilon^2)$ repetitions suffice for obtaining $(1 \pm \epsilon)$-factor estimates w.h.p.~of the sizes of $B(v, d)$ for all $v$ and all $d$.

Cohen proposes a modified Dijkstra's SSSP algorithm to find smallest rank vertices in each neighborhood.
Let $v_1, v_2, \ldots, v_n$ be the vertices of the graph in non-decreasing order of rank. 
Initiate Dijkstra's algorithm, first with source $v_1$, then with source $v_2$, and so on.
During the search with source $v_i$, if it is detected that for a vertex $u$, $d(u, v_j) \le d(u, v_i)$ for some $j < i$,
then the current search can be ``pruned'' at $u$.
This is because the vertex $v_j$ has ruled out $v_i$ from being the lowest ranked vertex in any of $u$'s neighborhoods.
In fact, this is true not just for $u$, but for all vertices whose shortest paths to $v_i$ pass through $u$.
Even though this algorithm performs $n$ SSSP computations, the fact that each search is pruned by the results of
previous searches makes the overall running time much less than $n$ times the worst case running time of
an SSSP computation.
In particular, by making critical use of the fact that the random vertex ranks induce a random permutation of the vertices,
Cohen is able to show that the algorithm runs in $O(m \log n + n \log^2 n)$ time, 
on $n$-vertex, $m$-edge graphs, w.h.p.

We don't know how to implement Cohen's algorithm, as is, efficiently in the $k$-machine model.
In particular, it is not clear how to take advantage of pruning that occurs in later searches while simultaneously taking advantage of the 
parallelism provided by the $k$ machines.
A naive implementation of Cohen's algorithm in the $k$-machine model is equivalent to $n$ different SSSP computations, which is too expensive. 
Below, in Algorithm \textsc{NbdSizeEstimates} (Algorithm \ref{alg:CohenEstimates}), we show that we can reduce Cohen's 
algorithm to a polylogarithmic number of SSSP computations provided
we are willing to relax the requirement that we find the smallest rank in each neighborhood. 

The goal of Algorithm \ref{alg:CohenEstimates} is to estimate $|B(v, d)|$ for all 
$v \in V$ and all \(d > 0\). 
In Step \ref{alg2:ChooseRanks}, each vertex $v \in V$ picks a rank uniformly at random 
from $[0, 1]$, which is rounded down to the closest value $(1+\epsilon')^i/n^2$ for some 
integer $i$ ($\epsilon'$ is suitably chosen in the algorithm). 
In Steps 5-7, in each iteration $i$, $0 \le i < \lceil \log_{1+\epsilon'}(n^2) \rceil$, 
we consider the set $T_i$ of vertices that have rounded rank equal to $(1+\epsilon')^i / n^2$ 
and solve an instance of the MSSP problem (see Lemma \ref{lemma:MSSP}) using the vertices 
in $T_i$ as sources. We repeat the algorithm $\lceil c \log n / (\epsilon')^2 \rceil$ times 
for a suitably chosen constant $c$, so that the neighborhood size estimates satisfy the 
property provided in Theorem \ref{thm:cohen} with high probability.

Notice that the algorithm's behavior is not well-defined if a rank falls in the range 
$[0, (1+\epsilon')/n^2)$ However, since ranks are chosen uniformly at random from $[0,1]$, 
the probability that the rank of a vertex falls in this range is $O(1/n^2)$. By union bound, 
no rank falls in the interval $[0, (1+\epsilon')/n^2]$ with probability at least $1-1/n$.
We condition the correctness proof of this algorithm on this high probability event.

\RestyleAlgo{boxruled}
\begin{algorithm2e}\caption{\textsc{NbdSizeEstimates}\((G, \epsilon)\)\label{alg:CohenEstimates}}
	$\epsilon' := \epsilon/(\epsilon + 4)$; $t = \lceil 2 \log_{1+\epsilon'} n \rceil$; $\ell := \lceil c \log n/(\epsilon')^2 \rceil$\\
  \For{\(j := 1, \ldots, \ell\)} {
	  \textbf{Local Computation.} Each machine $m_j$ picks a rank $\r(v)$, for each vertex $v \in H(m_j)$,
	  chosen uniformly at random from $[0, 1]$. Machine $m_j$ then rounds $\r(v)$ down to the closest 
	  $(1+\epsilon')^i/n^2$ for integer $i \ge 0$\label{alg2:ChooseRanks} \\
    \For{$i := 0, 1, \ldots, t-1$}{
      $T_i := \{v \in W \mid \r(v) = (1+\epsilon')^i/n^2\}$\\
	Compute a \((1+\epsilon)\)-approximate solution to MSSP using \(T_i\) as the set of sources \label{alg2:MSSP}; let $\tilde{d}(v, T_i)$ denote the computed approximate distances\\
	\textbf{Local Computation.} Machine $m_j$ stores $\tilde{d}(v, T_i)$ for each \(v \in H(m_j)\)\\
      }
  }
\end{algorithm2e}

\noindent
\textbf{Running time.} There are $\ell \cdot t$ calls to the subroutine solving MSSP. By Corollary \ref{theorem:SSSP},
each of these calls takes $\to(1/\mbox{poly}(\epsilon) \cdot n/k)$ rounds.
Since $\ell \cdot t = O((1/\mbox{poly}(\epsilon') \cdot \log^2 n)$, the overall round complexity of this algorithm
in the $k$-machine model is $\to(1/\mbox{poly}(\epsilon) \cdot n/k)$.

\noindent
\textbf{Answering queries.}
At the end of each iteration, each machine $m_j$ holds, for each vertex $v \in H(m_j)$, the sequence of 
distances, $\{\tilde{d}(v, T_i)\}_{i=0}^{t-1}$. Over $\ell$ repetitions, machine $m_j$ holds $\ell$ such sequences
for each vertex $v \in H(m_j)$.
Note that each distance $\tilde{d}(v, T_i)$ is associated with the rounded rank $(1 + \epsilon')^i/n^2$.
For any vertex $v \in V$ and real $d > 0$, let us denote the query ``What is the size of $B(v, d)$?'' by $Q(v, d)$.
To answer query $Q(v, d)$, we consider one of the $\ell$ sequences $\{\tilde{d}(v, T_i)\}_{i=0}^{t-1}$ and
find the smallest $i$, such that $\tilde{d}(v, T_i) \le d$, and return the rounded rank $(1+\epsilon')^i/n^2$.
To get an estimate that has low relative error, we repeat this over the $\ell$ sequences and compute the 
average \(\overline{R}\) of the
ranks computed in each iteration. 
The estimator is obtained by subtracting \(1\) from the reciprocal of \(\overline{R}\). 

The following lemma shows the correctness of Algorithm \ref{alg:CohenEstimates} in the sense
that even though we might not get an approximately correct answer to $Q(v, d)$,
the size \(|B(v, d)|\) is guaranteed to be ``sandwiched'' between the answers to 
two queries with nearby distances.
This guarantee is sufficient to ensure that the \textsc{RadiusComputation} Algorithm produces
approximately correct radii (see Section \ref{section:radiusComputation}).

\begin{Lemma}\label{lem:cohen}
Let $s$ denote $|B(v, d)|$ for some vertex $v$ and real $d > 0$. 
For any $0 < \epsilon < 1$, w.h.p., Algorithm \ref{alg:CohenEstimates} satisfies the following properties:
	\begin{itemize}
	\item for the query $Q(v, d/(1+\epsilon))$, the algorithm returns an answer that is at most $s(1+\epsilon)$.
	\item for the query $Q(v, d(1+\epsilon))$, the algorithm returns an answer that is at least $s/(1+\epsilon)$.
	\end{itemize}
\end{Lemma}
\begin{proof}
Fix a particular repetition $j$, $1 \le j \le \ell$, of the algorithm and a ranking of the vertices. 
Let $\r_j(v, d)$ denote the smallest rank in $B(v, d)$ in repetition $j$.
To answer query $Q(v, d/(1+\epsilon))$, the algorithm examines the sequence of approximate
distances $\{\tilde{d}(v, T_i)\}_{i=0}^{t-1}$, finds the smallest $i$ such that 
$\tilde{d}(v, T_i) \le d/(1+\epsilon)$, and uses $R_j := (1+\epsilon')^i/n^2$ as an approximation for 
$\r_j(v, d)$.
Since $\tilde{d}(v, T_i) \le d/(1+\epsilon)$ there is a vertex $u \in T_i$ such that $\tilde{d}(v, u) \le d/(1+\epsilon)$.
Since we compute a $(1+\epsilon)$-approximate solution to MSSP, the actual distance $d(v, u) \le d$. 
Thus the rank of $u$ is at least $\r_j(v, d)$ and therefore the rounded-rank of $u$ is at least $\r_j(v, d)/(1+\epsilon')$.
Since $u \in T_i$, the rounded-rank of $u$ is simply $R_j$ and so we get that $R_j \ge \r_j(v, d)/(1+\epsilon')$.

Over all $\ell$ repetitions, the algorithm computes the average $\overline{R}$ of the sequence $\{R_j\}_{j=1}^\ell$.
Letting $\overline{r}(v, d)$ denote the average of $\r_j(v, d)$ over all $\ell$ repetitions, we see that 
$\overline{R} \ge \overline{r}(v, d)/(1+\epsilon')$. From Theorem \ref{thm:cohen}, we know that w.h.p.~$\overline{r}(v, d) 
\ge (1 - \epsilon')/(1 + s)$. Combining these two inequalities, we get
\begin{eqnarray*}
\frac{1}{\overline{R}} & \le & \left(\frac{1+\epsilon'}{1-\epsilon'}\right) \cdot (s + 1)\\
\frac{1}{\overline{R}} - 1     & \le & \left(\frac{1+\epsilon'}{1-\epsilon'}\right) \cdot s + \left(\frac{2\epsilon'}{1-\epsilon'}\right)\\ 
	                       & \le & \left(\frac{1+3\epsilon'}{1-\epsilon'}\right) \cdot s\\ 
	                       & \le & (1+\epsilon) \cdot s.
\end{eqnarray*}
The second last inequality above follows from the fact $s \ge 1$, since $v \in B(v, d)$.
The last inequality follows from the setting $\epsilon' = \epsilon/(\epsilon + 4)$.

Now we consider query $Q(v, d \cdot (1+\epsilon))$.
Again, fix a repetition $j$, $1 \le j \le \ell$, of the algorithm and a ranking of the vertices. 
Let $u \in B(v, d)$ be a vertex with rank equal to $\r_j(v, d)$. 
We get two immediate implications: (i) the rounded-rank of $u$ is at most $\r_j(v, d)$ and
(ii) $\tilde{d}(v, u) \le d(1 + \epsilon)$. Together these imply that $R_j$, the approximate
rank computed by the algorithm in repetition $j$ is at most $\r_j(v, d)$.
Averaging over all $\ell$ repetitions we get that $\overline{R} \le \overline{r}(v, d)$.
Using Theorem \ref{thm:cohen}, we know that w.h.p.~$\overline{r}(v, d) 
\le (1 + \epsilon')/(1 + s)$. Combining these two inequalities, we that get $\overline{R} \le (1 + \epsilon')/(1+s)$.
This leads to
\begin{eqnarray*}
\frac{1}{\overline{R}} - 1 & \ge & \left(\frac{1 + s}{1 + \epsilon'}\right) - 1\\
                           & \ge & \left(\frac{1 - \epsilon'}{1 + \epsilon'}\right) \cdot s\\
                           & \ge & \frac{s}{1 + \epsilon}.
\end{eqnarray*}
The second last inequality follows from the fact that $s \ge 1$.
A little bit of algebra shows that $\epsilon' = \epsilon/(\epsilon + 4)$ implies
that $(1-\epsilon')/(1+\epsilon') \ge 1/(1+\epsilon)$ and the last inequality
follows from this.
\end{proof}


\subsection{Radius Computation Revisited}
\label{section:radiusComputation}

Having designed a $k$-machine algorithm that returns approximate neighborhood-size estimates
we restate the \textsc{RadiusComputation} algorithm (Algorithm \ref{alg:RC}) below.
\RestyleAlgo{boxruled}
\begin{algorithm2e}\caption{\textsc{RadiusComputation} Algorithm (Version 2)\label{alg:RC2}}
	\textbf{Neighborhood-Size Computation.} Call the \textsc{NbdSizeEstimates} algorithm (Algorithm \ref{alg:CohenEstimates}) to obtain \textit{approximate} neighborhood-size estimates $\tilde{q}_i(v)$
for all integers $i \ge 0$ and for all vertices $v$.\\
	\textbf{Local Computation.} Each machine $m_j$ computes $\tilde{r}_v$ locally, for all vertices $v \in H(m_j)$ using the formula $\tilde{r}_v := (1+\epsilon)^{t-1}$ where $t \ge 1$ is the smallest integer for which $\sum_{i = 0}^t \tilde{q}_i(v) \cdot ((1+\epsilon)^{i+1} - (1+\epsilon)^i) > f_v$.
\end{algorithm2e}

We show below that even though the computed neighborhood-sizes are approximate, in the sense
of Lemma \ref{lem:cohen}, the radii that are computed by the \textsc{RadiusComputation} algorithm (Version 2)
are a close approximation of the actual radii.
\begin{Lemma} \label{lem:facloc-guarantee}
	For every $v \in V$, $\frac{r_v}{(1+\epsilon)^3} \le \tilde{r}_v \le (1+\epsilon)^3 r_v$.
\end{Lemma}
\begin{proof}
  By Lemma \ref{lem:cohen}, we have the following bounds on \(\tilde{q}_i(v)\):
  \[\frac{1}{(1+\epsilon)}q_{i-1}(v)  \le \tilde{q}_i(v) \le (1 + \epsilon) q_{i+1}(v)\]

  Similar bounds will apply for the terms \(((1+\epsilon)^{i+1} - (1+\epsilon)^{i})\tilde{q}_i(v)\). Adding the respective inequalities for these terms, yields the following inequality:
  \[\sum_{i = 0}^{t-1} ((1+\epsilon)^i - (1+\epsilon)^{i-1}) q_{i-1}(v) \le \sum_{i = 0}^{t-1} ((1+\epsilon)^{i+1} - (1+\epsilon)^{i}) \tilde{q}_i(v) \le \sum_{i = 0}^{t-1} ((1+\epsilon)^{i+2} - (1+\epsilon)^{i+1}) q_{i+1}(v).\]

  Now we obtain the following bound using similar arguments as in Lemma \ref{lemma:approxRadii1}:
  \[\alpha(v, (1+\epsilon)^{t-2}) \le \sum_{i = 0}^{t-1} \tilde{q}_i(v) \cdot ((1+\epsilon)^{i+1} - (1+\epsilon)^i) \le \alpha(v, (1+\epsilon)^{t+1}).\]
  
  This means that there must exist a value $r_v \in [ (1+\epsilon)^{t-3}, (1+\epsilon)^{t+1}]$ such that $\alpha(v, r_v) = f_v$. The lemma follows since $\tilde{r}_v = (1 + \epsilon)^{t-1}$.
\end{proof}

\subsection{Implementing the Greedy Phase}

Referring to the two phases in the MP Algorithm (Algorithm \ref{alg:MP}), we have now completed the implementation of the 
Radius Computation Phase in the $k$-machine model. Turning to the Greedy Phase, we note that discretizing the radius values results in $O(\log_{1+\epsilon} n)$
distinct values. If we can efficiently process each batch of vertices with the same (rounded) radius in the $k$-machine model,
that would yield an efficient $k$-machine implementation of the Greedy Phase as well.
Consider the set $W$ of vertices with (rounded) radius $\tilde{r}$.
Note that a set $I \subseteq W$ is opened as facilities by the Greedy Phase iff $I$ satisfies two properties: (i) for any two vertices
$u, v \in I$, $d(u, v) > 2 \tilde{r}$ and (ii) for any $w \in W \setminus I$, $d(w, I) \le 2 \tilde{r}$.
Thus the set $I$ can be identified by computing a \textit{maximal independent set (MIS)} in the graph $G_{\tilde{r}}[W]$,
where $G_{\tilde{r}}$ is the graph with vertex set $V$ and edge set $E_{\tilde{r}} = \{\{u, v\} \mid u, v \in V, d(u, v) \le \tilde{r}\}$.
($G_{\tilde{r}}[W]$ denotes the subgraph of $G_{\tilde{r}}$ induced by $W$.)

The well-known distributed MIS algorithm of Luby \cite{Luby1986} runs in $O(\log n)$ rounds w.h.p.~and it can be easily implemented in the $k$-machine model in
$O(n/k \cdot \log n)$ rounds.
However, Luby's algorithm assumes that the graph on which the MIS is being computed is provided explicitly. This is not possible here because
explicitly providing the edges of a graph $G_{d}$ would require pairwise-distance computation, which we're trying to avoid.
Another problem with using Luby's algorithm is that it uses randomization, where the probabilities of certain events depend on vertex-degrees.
The degree of a vertex $v$ in $G_d[W]$ is exactly $|B(v, d) \cap W|$ and this is the quantity we would need to estimate.
Unfortunately, the correctness guarantees for Algorithm \ref{alg:CohenEstimates} proved in Lemma \ref{lem:cohen} are not strong enough
to give good estimates for $|B(v, d) \cap W|$.
We deal with these challenges by instead using the \textit{beeping model} MIS algorithm of Afek et al.~\cite{AfekABHBBScience2011},
which is quite similar to Luby's algorithm except that it does require knowledge of vertex-degrees.
In Luby's algorithm vertices ``mark'' themselves at random as candidates for joining the MIS. After this step, if a marked vertex $v$ detects that a neighbor
has also marked itself, then $v$ ``backs off.''
In the current setting, this step would require every marked vertex $v$ to detect if there is \textit{another} marked vertex within distance
$d$. We use ideas from Thorup \cite{Thorup2001} to show that this problem can be solved using $O(\log n)$ calls to a subroutine that solves
\textsc{ExclusiveMSSP} (Lemma \ref{lemma:ExclusiveMSSP}).
In Luby's algorithm marked vertices that do not back off, join the MIS (permanently).
Then, any vertex $v$ that has a neighbor who has joined the MIS will withdraw from the algorithm.
Determining the set of vertices that should withdraw in each iteration requires a call to an MSSP subroutine.
Because the calls to the \textsc{ExclusiveMSSP} and MSSP subroutines return only 
approximate shortest path distances, what Algorithm \ref{alg:DistdMIS} computes is a 
relaxation of an MIS, that we call $(\epsilon,  d)$-approximate MIS.

\begin{Definition}[$(\epsilon, d)$-approximate MIS]
	For an edge-weighted graph $G = (V, E)$, and parameters $d, \epsilon > 0$, an $(\epsilon, d)$-approximate MIS is a subset $I \subseteq V$ such that
	\begin{enumerate}
		\item For all distinct vertices $u, v \in I$, $d(u, v) \ge \frac{d}{1+\epsilon}$.
		\item For any $u \in V \setminus I$, there exists a $v \in I$ such that $d(u, v) \le d \cdot (1+\epsilon)$.
	\end{enumerate}
\end{Definition}

\RestyleAlgo{boxruled}
\begin{algorithm2e}\caption{\textsc{ApproximateMIS}\((G, W, d, \epsilon)\)\label{alg:DistdMIS}}
	Each machine $m_j$ initializes \(U_j := \emptyset\) \\
	\tcc{Let $W_j$ denote $W \cap H(m_j)$.}
	\For{$i := 0, 1, \ldots, \lceil \log n \rceil$}{
		\For{$\lceil c \log n \rceil$ iterations}{
			Each machine \(m_j\) marks each vertex $v \in W_j$ with probability $2^i/n$ \\
			\tcc{Let \(R_j \subset W_j\) denote the set of marked vertices hosted by $m_j$, let $R := \cup_{j=1}^k R_j$} 
			Solve an instance of the \textsc{ExclusiveMSSP} problem using $R$ as the set of sources (see Lemma \ref{lemma:ExclusiveMSSP}) to obtain $(1+\epsilon)$-approximate distances $\tilde{d}$ \label{alg1:MSSP1}\\
			Each machine $m_j$ computes $T_j := \{v \in R_j \mid \tilde{d}(v, R \setminus \{v\}) > d\}$ \\
    			Each Machine $m_j$ sets $U_j := U_j \cup T_j$ \\
			\tcc{Let $T := \cup_{j=1}^k T_j$}
			Solve an instance of the \textsc{MSSP} problem using $T$ as the set of sources (see Lemma \ref{lemma:MSSP}) to obtain $(1+\epsilon)$-approximate distances $\tilde{d}$ \label{alg1:MSSP2} \\
			Each machine $m_j$ computes $Q_j = \{v \in W_j \mid \tilde{d}(v, T) \le d\}$ \\
			Each machine $m_j$ sets $W_j := W_j \setminus (T_j \cup Q_j)$
		}
	}
  	\textbf{return} \(U := \cup_{j = 1}^k U_j\)
\end{algorithm2e}

The algorithm consists of \(\lceil \log n \rceil\) stages and in each Stage \(i\), 
we run a Luby-like MIS algorithm for \(\lceil c\log n \rceil\) iterations 
(for some constant \(c > 0\)) with fixed marking probability which we double in each stage. 
In each iteration of the two for loops, the set \(R_j\) is the set of marked vertices in 
machine \(m_j\). 
The machines solve an \textsc{ExclusiveMSSP} instance in Step \ref{alg1:MSSP1} with all 
marked vertices to ensure that marked vertices that are within approximated distance 
\(d\) of each other back-off. The marked vertices in machine \(m_j\) that do not back-off 
(i.e., vertices in \(T_j\)) join the MIS (\(U\)). The machines then solve an instance of 
the \textsc{MSSP} problem in Step \ref{alg1:MSSP2} to remove the vertices that within 
approximate distance \(d\) from the vertices in the MIS. We formalize the correctness of 
Algorithm \ref{alg:DistdMIS} in the following Lemma.

\begin{Lemma} \label{lem:DistdMIS}
For a given set $W \subseteq V$, Algorithm \ref{alg:DistdMIS} finds an $(O(\epsilon), d)$-approximate MIS $I$ 
of $G[W]$ whp in \(\tilde{O}(n/k)\) rounds.
\end{Lemma}
\begin{proof}
  We first bound the running time of the algorithm. The double nested
  loop runs for $O(\log^2 n)$ iterations. In each iteration, Steps
  \ref{alg1:MSSP1} and \ref{alg1:MSSP2} run in $\to(n/k)$ rounds via
  Lemmas \ref{lemma:ExclusiveMSSP} and \ref{lemma:MSSP} respectively and
  all other steps are local computations. This means that the overall
  running time is \(\tilde{O}(n/k)\).

  By the analysis in \cite{AfekABHBBScience2011} and the guarantees
  provided by the solution to \textsc{ExclusiveMSSP}, no two vertices in
  $W$ at distance at most $d/(1+\epsilon)$ end up in $T$.  Similarly, by
  the analysis in \cite{AfekABHBBScience2011} and the guarantees
  provided by the solution to \textsc{MSSP}, every vertex in $W
  \setminus U$ is at distance at most $d(1+\epsilon)$ from $U$.  Thus
  $U$ is an $(O(\epsilon), d)$-approximate MIS and it is computed in the
  $k$-machine model in $\to(n/k)$ rounds.
\end{proof}
  
\subsection{Putting It All Together}  \label{subsec:primaldual}

\RestyleAlgo{boxruled}
\begin{algorithm2e}\caption{\(\beta\)-\textsc{MettuPlaxton}\((G)\)\label{alg:MettuPlaxtonPhase2}}
  \tcc{Start Phase 1 of the \(\beta\)-MP algorithm}
  Call the \textsc{RadiusComputation} algorithm Version 2 (Algorithm \ref{alg:RC2}) to compute  approximate radii. \label{alg4:radius-computation}\\
  \tcc{Start Phase 2 of the \(\beta\)-MP algorithm}
  Let \(S = \emptyset\) \\
  \For{\(i = 0, 1, 2, \dots\)}{
    Let \(W\) be the set of vertices \(w \in V\) across all machines with \(\tilde{r}_w = \tilde{r} = (1 + \epsilon)^i\) \\
    Using Lemma \ref{lemma:MSSP}, remove all vertices from \(W\) within distance \(2(1 + \epsilon)^2\cdot \tilde{r}\) from \(S\) \label{alg4:remove-vertices}\\
    \(I \leftarrow \textsc{ApproximateMIS}(G, W, 2(1 + \epsilon)^3 \cdot \tilde{r}, \epsilon)\) \\
    \(S \leftarrow S \cup I\)
  }
  \textbf{return} \(S\)

\end{algorithm2e}

Our $k$-machine model algorithm for \facloc\ is shown in Algorithm \ref{alg:MettuPlaxtonPhase2}.
We could analyze the algorithm as done in \cite{Thorup2001} to show the constant approximation guarantee. However, we want to use this algorithm for obtaining a $p$-median algorithm in the next section. Therefore, we take an approach similar to \cite{JainV01} and \cite{ArcherRSESA2003}, to show a stronger approximation guarantee in Lemma \ref{lem:mp-beta} (see appendix). 
We require several claims, which are along the lines of those in Thorup \cite{Thorup2001}, 
and Archer et al.~\cite{ArcherRSESA2003}. The details are technical and since they largely appear in 
Thorup \cite{Thorup2001} and Archer et al.~\cite{ArcherRSESA2003}, they are deferred to the 
appendix. Finally, using Lemma \ref{lem:mp-beta}, we get the following theorem.

\begin{theorem} \label{thm:FacLocGuarantee}
In \(\tilde{O}(n/k)\) rounds, whp, Algorithm \ref{alg:MettuPlaxtonPhase2} finds a factor \(3 + O(\epsilon)\) approximate solution \(S\) to the facility location problem. 
Furthermore, if \(F\) is the total facility cost of the algorithm's solution, 
\(C\) is the total connection cost of the algorithm's solution, 
\(OPT\) is the optimal solution cost, and $\beta \in [1, 3/2]$ then \(C + 2\beta F \le 3 (1 + \epsilon) \sum_{j}{v_j} \le 3(1+\epsilon)OPT\)
\end{theorem}
\begin{proof}
	Algorithm \ref{alg:MettuPlaxtonPhase2} consists of two phases which correspond to the Radius Computation and Greedy Phases of the MP algorithm (Algorithm \ref{alg:MP}). We bound the running time of both these phases. There are at most $O(\log_{1 + \epsilon}{nN}) = O(\log nN) = O(\log n)$ possible values of \(i\) and hence at most $O(\log n)$ iterations in the two phases of Algorithm \ref{alg:MettuPlaxtonPhase2} (where \(N = \mbox{poly}(n)\) is the largest edge weight). In each iteration of Algorithm \ref{alg:RC2} consists of a call to Algorithm \ref{alg:CohenEstimates} which runs in \(\tilde{O}(n/k)\) rounds and hence Phase \(1\) of Algorithm \ref{alg:MettuPlaxtonPhase2}) requires \(\tilde{O}(n/k)\) rounds. Each iteration in Phase \(2\) Algorithm \ref{alg:MettuPlaxtonPhase2} takes \(\tilde{O}(n/k)\) rounds therefore we conclude that the overall running time is \(\tilde{O}(n/k)\) rounds.
	
	As for the approximation guarantee, we note by Lemma \ref{lem:mp-beta}, we get that for each vertex $j \in V$, we have shown that there exists an opened facility $c(j) \in S$ such that $(3 + \epsilon)\cdot v_j \ge d(j, c(j)) + \beta s_j$ which gives the desired guarantee. Finally, we note that the cost of any feasible dual solution is a lower bound on the optimal cost. Then, by setting $\beta \in [1, 3/2]$ appropriately, the theorem follows.
\end{proof}

\section{A \(p\)-median algorithm}
\label{section:pMedian}
In this section, we describe an $\tilde{O}(n/k)$ round algorithm for the $p$-median problem. We will follow the randomized rounding algorithm of Jain and Vazirani \cite{JainV01} which shows an 
interesting connection between $p$-median and uniform facility location problems. As observed in \cite{JainV01}, the similarities between the linear programming formulations of the uniform facility location problem, and the $p$-median problem can be exploited to obtain an $O(1)$ approximation algorithm for the $p$-median problem, if one has a subroutine that returns an $O(1)$ approximation for the uniform facility location problem, with a specific property. This is summarized in the following lemma.

\begin{Lemma}[Modified from \cite{JainV01}]\label{lem:jv}
Let $\mathcal{A}$ be a polynomial time uniform facility location algorithm that takes the facility opening cost $z$ as input and returns a solution such that, 
$C + \mu \cdot F z \le \mu \cdot OPT$ where $C$ is the total connection cost, 
$F$ is the number of facilities opened by the algorithm, and $OPT$ is the optimal solution cost. Then there exists a randomized $p$-median algorithm $\mathcal{A}'$ that returns a solution with expected cost at most $2\mu$ times the optimal $p$-median cost. 
\end{Lemma}

Note that the facility location algorithm described in Section \ref{section:facilityLocation} 
returns a solution satisfying the guarantee in Lemma \ref{lem:jv} (cf. Theorem \ref{thm:FacLocGuarantee}). All that we need to show is that the randomized rounding algorithm can be efficiently implemented in the $k$-machine model. In the following sections, we first describe the sequential randomized algorithm $\mathcal{A}'$ \cite{JainV01}, and then discuss how to implement it in $k$-machine model.

\subsection{The Sequential Algorithm} \label{subsubsec:seq-randalgo}

Let $c_{\max}$ and $c_{\min}$ be the maximum and minimum inter-point distances respectively. Using a Facility Location algorithm that has the guarantee of Lemma \ref{lem:jv}, we perform binary search on the facility opening cost $z$ in the range $[0, n\cdot c_{\max}]$. If we come across a solution $A'$ such that $|A'| = p$, then we have a $\mu$-approximate solution and we stop. Otherwise, we find two solutions $A$ and $B$, such that $|A| < p < |B|$, with $z_A - z_B \le c_{\min}/(12n^2)$, where $z_A$ and $z_B$ are the facility opening costs corresponding to the solutions $A$ and $B$ respectively. Let $p_1 = |A|$ and $p_2 =|B|$. We now obtain a solution $C$ from $A$ and $B$, such that $|C| = p$.

Construct the set $B' \subseteq B$ as follows. Starting with an empty set, for each vertex in $A$, add the closest vertex in $B$ to $B'$, breaking ties arbitrarily. If at this point, $|B'| < p_1$, add arbitrary vertices from $B \setminus B'$ to $B'$ until $|B'| = p_1$. Set $C = A$, with probability $a$, and $C = B'$ with probability $b$, where $a = \frac{p_2-p}{p_2 - p_1}, b = \frac{p-p_1}{p_2 - p_1}$. Now, pick a set of $p-p_1$ vertices from $B \setminus B'$, and add it to $C$. It is clear that $|C| = p$, and this is the claimed solution with expected cost $2\mu$ times that of the optimal $p$-median cost.

\subsection{Implementation in the \texorpdfstring{$k$}{k}-machine Model}

In order to implement the sequential algorithm in the $k$-machine model, we will assign a special machine (say the machine with the smallest ID), which executes the key steps of the sequential algorithm. For convenience, we refer to this machine as $M_1$. First, each machine sends the weights of minimum and maximum weight edges incident on any of the vertices hosted by it to $M_1$. This allows $M_1$ to figure out the smallest edge weight \(w_{min}\) and the largest edge weight \(w_{max}\) in the input graph and it sets $c_{min} = w_{min}$ and $c_{max} = n \cdot w_{max}$ (which is a crude polynomial upper bound). The machines perform binary search on the facility opening cost to obtain two solutions $A$, and $B$ by using Algorithm \ref{alg:MettuPlaxtonPhase2} (modified appropriately to take facility opening cost as input parameter). We assume that each machine knows the subsets of the vertices hosted by it that belong to $A$ and $B$ respectively.

Now, we show how the machines identify the set $B' \subseteq B$ in $\tilde{O}(n/k)$ rounds. Using Lemmas \ref{lemma:MSSP} and \ref{lemma:ExclusiveMSSP} with \(T = B\), for each vertex in  $A$, we determine the approximately closest vertex from $B$ in $\to(n/k)$ rounds, and let $B''$ be this set. At this point, each machine also knows which of its vertices belongs to $B''$. In $O(1)$ rounds, each machine sends the number of its vertices belonging to $A, B$, and $B''$, to $M_1$. If $M_1$ discovers that $|B''| < p_1$, then it decides arbitrary $p_1 - |B''|$ vertices from $B$, and informs the respective machines to mark those vertices as belonging to $B'$, and update the counts accordingly. This takes $\tilde{O}(n/k)$ rounds.

Now, $M_1$ locally determines whether $A$ or $B'$ will be included in the solution set $C$ (with probability $a$ and $b$ respectively) and informs all other machines. Note that $M_1$ knows the number of vertices in $B \setminus B'$ that belong to each of the machines so it can sample $p - p_1$ vertices in the set $B'' \subseteq B \setminus B'$ as follows. For a machine $M_j$, $M_1$ sends it the number $t_j$ which is the number of vertices from $B \setminus B'$ hosted by $M_j$ that are chosen by $M_1$ uniformly at random to be in $B''$. Finally, each machine $M_j$ chooses a set of $t_j$ vertices uniformly at random from the set $B \setminus B'$ that it hosts. It is easy to see that this procedure guarantees that each vertex from the set $B \setminus B'$ has probability $b$ of getting chosen in the set $B''$. The set $C \gets C \cup B''$ is the final solution.


At this point, each machine knows the subset of $C$ that is hosted by it. We use Lemmas \ref{lemma:MSSP} and \ref{lemma:ExclusiveMSSP} to identify for each vertex $u \in V$, the approximately closest vertex $v \in C$ in $\tilde{O}(n/k)$ rounds. In additional $\tilde{O}(n/k)$ rounds, $M_1$ can compute the approximate cost of the solution. Note that in this step, and while computing $B'$, we use an $(1+\epsilon)$-approximate SSSP algorithm, instead of an exact SSSP algorithm. 
Using this fact in the analysis of \cite{JainV01}, it can be shown that this does not increase the expected cost of the solution by more than an $O(\epsilon)$ factor. We omit the details. Thus, the solution obtained by our algorithm has cost at most $6 + O(\epsilon)$ times the optimal solution with high probability. Finally, setting the value of $\epsilon$ for the facility location algorithm appropriately yields the following theorem.
\begin{theorem}
	For any constant $\epsilon > 0$, there exists a randomized algorithm to obtain a $6 + \epsilon$ factor approximation to the $p$-median problem in the $k$-machine model in $\tilde{O}(n/k)$ rounds w.h.p.
\end{theorem}
 \section{A $p$-center algorithm}
\label{section:pCenter}

In this section, we describe a constant factor approximation algorithm for the $p$-center problem. It is a well-known that (see for example \cite{Gonzalez1985}), if $d^*$ is an optimal $p$-center cost, then any distance-$2d^*$ MIS is a $2$-approximation for the $p$-center. But since we do not know how to compute a distance-$d$ MIS efficiently in the $k$-machine model, we show in the following Lemma that an $(\epsilon, 2 \cdot (1+\epsilon)d^*)$-approximate MIS suffices to get an $O(1)$-approximation.

\begin{Lemma} \label{lem:apx-pcenter}
	For a graph $G = (V, E)$, if $d^*$ is an optimal $p$-center cost, then any $(\epsilon, 2(1+\epsilon)d^*)$-approximate MIS is an $2(1+\epsilon)^2$ approximation.
\end{Lemma}
\begin{proof}
	Let $O = \{o_1, o_2, \cdots, o_p\} \subseteq V$ be an optimal $p$-center solution (we assume without loss of generality that $O$ contains exactly $p$ centers). Define a partition $\{V_i\}$ of the vertex set $V$, by defining the set $V_i$ for each $o_i \in O$ as follows. For each $o_i \in O$, let $V_i \subseteq V$ be the set of vertices, for which $o_i$ is the closest center in $O$. Here we break ties arbitrarily, so that each vertex appears in exactly one of the sets $V_i$. Note that if $v \in V_j$ for some $j$, then $d(v, o_j) = d(v, O) \le d^*$.
	
	Now let $I \subseteq V$ be any $(\epsilon, 2(1+\epsilon)d^*)$-approximate MIS. We first show that $I$ is feasible, i.e. $|I| \le p$, by showing that for any $i \in \{1, 2, \cdots, p\}$, $|V_i \cap I| \le 1$. Assume this is not the case, i.e. for some $i$, there exist distinct $v_1, v_2 \in V_i \cap I$. But this implies that $d(v_1, v_2) \le d(v_1, o_i) + d(o_i, v_2) \le 2d^*$, which is a contradiction to the fact that $I$ is an $(\epsilon, 2(1+\epsilon)d^*)$-approximate MIS.
	
	Finally, the approximation guarantee follows from the definition of an approximate MIS -- for any $v \in V$, there exists an $u \in I$ such that $d(u, v) \le 2(1+\epsilon)^2 d^*$.
\end{proof}

Although we do not know the optimal $p$-center cost $d^*$ we can find it by doing a binary search to get the largest $d$ such that an $(\epsilon, 2(1+\epsilon)d)$-approximate MIS has size at most $p$. There are at most $O(\log n)$ iterations of the binary search because of our assumption that the distances bounded by $\mbox{poly}(n)$. This along with Lemma \ref{lem:DistdMIS} gives us the following theorem.

\begin{theorem}
	For any constant $\epsilon > 0$, there exists a randomized algorithm to obtain a $(2+\epsilon)$-factor approximation to the $p$-center problem in the $k$-machine model in $\tilde{O}(n/k)$ rounds w.h.p.
\end{theorem} 

\section{Conclusions}
\label{section:conclusions}

This paper initiates the study of clustering problems in the $k$-machine model
and presents near-optimal (in rounds) constant-factor approximation algorithms for
these problems.
The near-optimality of our algorithms is established via almost-matching lower bounds on
on the number of rounds needed to solve these problems in the $k$-machine model.
However, the lower bounds critically depend a certain assumption regarding how the output
of the clustering algorithms is to be represented.
Specifically, we require that every machine with an open facility knows all clients
connecting to that facility.
This requirement forces some machines to learn a large volume of information distributed
across the network and this leads to our lower bounds.

We could alternately, impose a rather ``light weight'' output requirement and, for example,
require each machine with an open facility to simply know the \textit{number} of clients connecting to
it or the aggregate connection cost of all the clients connecting to it.
(Of course, independent of this change, the output requires that each client know the facility it connects to.)
So the main open question that follows from our work is whether we can design
optimal $k$-machine algorithms under this relaxed output requirement.
$\Omega(n/k^2)$ lower bounds do not seem difficult to prove in this setting, but to obtain
$\tilde{O}(n/k^2)$-round constant-approximation algorithms seems much harder.
Alternately, can we prove stronger lower bounds even in this, more relaxed, setting?

\bibliographystyle{plain}
\bibliography{bibliography} 

\appendix
\section{Technical Proofs from Section \ref{subsec:primaldual}} \label{appendix:beta-mp-proofs}
In this appendix, we give some of the technical proofs required to prove Lemma \ref{lem:mp-beta}, and then give its proof.

Throughout this section, we condition on the event that the outcome of all the randomized algorithms is as expected (i.e. the ``bad'' events do not happen). Note that this happens with w.h.p. We first need the following facts along the lines of \cite{Thorup2001}. 

\begin{Lemma}[Modified From Lemma 8 Of \cite{Thorup2001}] \label{lem:lem8}
	There exists a total ordering $\prec$ on the vertices in $V$ such that $u \prec v \implies \rdown_u \le (1+\epsilon) \cdot \rdown_v$, and $v$ is added to $S$ if and only if there is no previous $u \prec v$ in $S$ such that $d(u, v) \le 2(1+\epsilon)^2 \rdown_v$.
\end{Lemma}
\begin{proof}[Proof Sketch]
	The ordering is obtained by listing in each iteration, the vertices in $I$ that are included in $S$ before the rest of the vertices of $W$. Note that the extra $(1+\epsilon)^2$ factor appears because of the definition of $(\epsilon, d)$-approximate MIS.
\end{proof}

\begin{Claim}[Modified From Claim 9.2 Of \cite{Thorup2001}] \label{cl:cl9.2}
	For any two distinct vertices $u, v \in S$, we have that $d(u, v) \ge 2(1+\epsilon)^3 \cdot \max\{\rdown_u, \rdown_v\}$.
\end{Claim}
\begin{proof}[Proof Sketch]
	Without loss of generality, assume that $u \prec v$, so $r_u \le (1+\epsilon) \cdot r_v$. Now the claim follows from lemma \ref{lem:lem8}, and the definition of $(\epsilon, d)$-approximate MIS.
\end{proof}

In the rest of the section, we follow the primal-dual analysis of \cite{ArcherRSESA2003}, again with necessary modifications arising from various approximations. For completeness, we state the primal and dual LP relaxations below. We reserve the subscript \(i\) for facilities and \(j\) for clients. Note that in our case, \(i, j \in V\).
\begin{align*}
  \text{min. } \sum_{i} f_i y_i + &\sum_{i, j} d(i, j) \cdot x_{ij} & \text{max. } \sum_{j} v_j& \\
  \text{s.t.} \quad \sum_{i} x_{ij} &= 1 \quad \forall j & \text{s.t.} \quad \sum_{j} w_{ij} - f_i &\le 0 \quad \forall i \\
  y_{i} - x_{ij} &\ge 0 \quad \forall i, j & v_j - w_{ij} - d(i ,j) &\le 0 \quad \forall i, j \\
  y_i, x_{ij} & \ge 0 \quad \forall i, j & v_j, w_{ij} & \ge 0 \quad \forall i, j
\end{align*}

Let $\beta \ge 1$ be the parameter that is used in the Facility Location algorithm. Set $w_{ij} = \frac{1}{\beta} \max\{0, \rsim_i - \dup(i, j)\}$. Say that $j$ contributes to $i$ if $w_{ij} > 0$. Then, set $v_j = \min_{i \in V} d(i, j) + w_{ij}$. It is easy to see that the $v$ and $w$ values are dual feasible.

Define for each $j \in V$, $s_j = w_{ij}$ if there exists an $i \in S$ with $w_{ij} > 0$ and $0$ otherwise. Note that $s_j$ is uniquely defined, if it is not zero. This is because of the fact that the balls $B(v, \rdown_v)$ and $B(u, \rdown_u)$ are disjoint using Claim \ref{cl:cl9.2}. Also note that $f_i = \sum_{j \in V} w_{ij}$, therefore, $\sum_{i \in S} f_i = \sum_{j \in V} s_j$.

For $j \in V$, call the facility $i \in P$ that determines the minimum in $v_j = \min_{i \in V} d(i, j) + w_{ij}$, the \emph{bottleneck} of $j$. We say that a facility (or a vertex) is \emph{closed} if it does not belong to the set $S$, and it is opened otherwise. Furthermore, we say that a facility $v \in S$ \emph{caused} another facility $u \notin S$ was closed, if at the time $u$ was removed in the Algorithm \ref{alg:MettuPlaxtonPhase2}, Line \ref{alg4:remove-vertices}, $d(u, v) \le 2(1+\epsilon)^i$. Before showing the approximation guarantee, we need the following four lemmas. (cf. Lemmas 1-4 from \cite{ArcherRSESA2003})

\begin{Lemma} \label{lem:duallemma1}
	For any $i, j \in V$, we have that 
	$\tilde{r}_i \le (1+\epsilon)^3 \cdot (\beta w_{ij} + d(i, j))$. Furthermore, if for some $i, j \in V, w_{ij} > 0$, then $\tilde{r}_i \ge \frac{1}{(1+\epsilon)^3} (\beta w_{ij} + d(i, j))$.
\end{Lemma}
\begin{proof}
 	We have that $\beta w_{ij} \ge \rsim_i - \dup(i, j)$. Now using the appropriate upper and lower bounds from lemma \ref{lem:facloc-guarantee} for $\rsim_i$ to get the desired inequality.
\end{proof}

\begin{Lemma} \label{lem:duallemma2}
	If $\beta \le 3$, and if $i$ is a bottleneck for $j$, then $3(1+\epsilon)^3 v_j \ge 2\tilde{r}_i$.
\end{Lemma}
\begin{proof}
 	\begin{align*}
 	v_j &= d(i, j) + w_{ij} \tag{Since $i$ is the bottleneck for $j$}
 	\\\implies 3(1+\epsilon)^4 v_j &\ge (1+\epsilon)^4 (2 \cdot d(i, j) + 2 \cdot \beta) \tag{Using the fact that $\beta \le \frac{3(1+\epsilon)^4}{2}$}
 	\\&\ge 2 \cdot r_i \tag{Using Lemma \ref{lem:duallemma1}.}
 	\end{align*}
\end{proof}

\begin{Lemma} \label{lem:duallemma3}
	If an open facility is a bottleneck for $j$, then $j$ cannot contribute to any other open facility.
\end{Lemma}
\begin{proof}
 	Suppose $i' \in S$ is $j$'s bottleneck. Also assume that $j$ contributes to another $i \in S$, i.e. $w_{ij} > 0$. Using triangle inequality, we have that $d(i, i') \le d(i, j) + d(i', j) \le \dup(i, j) + d(i', j) < \rsim_i + d(i', j)$. In the last inequality, we use the fact that $w_{ij} > 0$, which means that $\dup(i, j) < \rsim_i$. Now there are two cases, depending on whether $w_{i'j} > 0$ or $w_{i'j} = 0$.  

 	In the first case, if $w_{i'j} > 0$, then again using similar reasoning, we have that $\dup(i', j) \le \rsim_{i'}$. However, this implies that $d(i, i') < \rsim_i + \rsim_{i'} \le 2(1+\epsilon)^2 (\rdown_i + \rdown_{i'}) \le 2(1+\epsilon^2) \max\{\rdown_i, \rdown_{i'}\}$, which is a contradiction to Claim \ref{cl:cl9.2}.

 	In the second case, $w_{i'j} = 0$. However, since $i'$ is also a bottleneck for $j$, this implies that $v_j = \min_{i \in V} d(i, j) = d(i', j)$. That is, $i'$ is the closest vertex to $j$, i.e. $d(i, j) \le d(i', i)$. However, this implies $d(i, i') \le 2d(i', j) \le 2\dup(i', j) < 2\rsim_{i'} \le 2(1+\epsilon)^2 \rdown_{i'}$, which is again a contradiction to Claim \ref{cl:cl9.2}.
\end{proof}

\begin{Lemma} \label{lem:duallemma4}
	If a closed facility $i \notin S$ is a bottleneck for $j \in V$, and $k \in S$ is the open facility that caused $i$ to close, then 
	$\max \{2\beta, 3\} \cdot (1+\epsilon)^7 \cdot v_j \ge d(k, j)$.
\end{Lemma}
\begin{proof}
 	\begin{align*}
 	d(k, j) &\le d(k, i) + d(i, j) \tag{Triangle inequality}
 	\\&\le 2(1+\epsilon)^3 \rdown_i + d(i, j) \tag{$k$ caused $i$ to close, so using Lemma \ref{lem:lem8}.}
 	\\&\le 2(1+\epsilon)^3 r_i + d(i, j)
 	\\&\le 2(1+\epsilon)^3 \cdot (1+\epsilon)^4 (\beta w_{ij} + d(i, j)) + d(i, j) \tag{Using Lemma \ref{lem:duallemma1}.}
 	\\&\le 2\beta (1+\epsilon)^7 w_{ij} + (2(1+\epsilon)^7 + 1) \cdot d(i, j)
 	\\&\le \max \{2\beta (1+\epsilon)^7, 2(1+\epsilon)^7 + 1 \} \cdot v_j \tag{Since $i$ is the bottleneck for $j$}
 	\\&\le \max \{2\beta, 3\} \cdot (1+\epsilon)^7 \cdot v_j
 	\end{align*}
\end{proof}

\begin{Lemma} \label{lem:duallemma4}
	If a closed facility $i \notin S$ is a bottleneck for $j \in V$, and $k \in S$ is the open facility that caused $i$ to close, then 
	$\max \{2\beta, 3\} \cdot (1+\epsilon)^7 \cdot v_j \ge d(k, j)$.
\end{Lemma}

\begin{proof}
	\begin{align*}
		d(k, j) &\le d(k, i) + d(i, j) \tag{Triangle inequality}
		\\&\le 2(1+\epsilon)^2 r_i + d(i, j) \tag{Using Lemma \ref{lem:lem8}}.
		\\&\le 2(1+\epsilon)^4 \tilde{r}_i + d(i, j) \tag{Using Lemma \ref{lem:facloc-guarantee}}
		\\&\le 2(1+\epsilon)^7 (\beta w_{ij} + d(i,j)) + d(i, j) \tag{Using Lemma \ref{lem:duallemma1}}
		\\&= \max\{2\beta, 3\} \cdot (1+\epsilon)^7 \cdot v_j \tag{Because $i$ is the bottleneck for $j$}
	\end{align*}
\end{proof}

We are finally ready to prove the main guarantee of the modified MP-$\beta$ algorithm, as in \cite{ArcherRSESA2003}. The basic idea is to show that $(3+O(\epsilon))$ times the dual variable $v_j$ pays for the distance traveled by $j$, as well as, $\beta s_j$, which is a part of the facility opening cost. We formalize this in the following lemma.

\begin{Lemma} \label{lem:mp-beta}
	For any vertex $j \in V$, there exists a facility $c(j) \in S$ such that $3 (1+O(\epsilon)) v_j \ge d(j, c(j)) + \beta s_j$.
\end{Lemma}
\begin{proof}
	Consider a vertex $j \in V$. We prove the theorem by doing a careful case analysis. 
	
	\textbf{Case 1.} Some open facility $i \in S$ is the bottleneck for $j$. Connect $j$ to $i$.
	\\If $d(i, j) \le \rsim_i$, we have that $0 < w_{ij} = s_j$. Also, $v_j = d(i, j) + s_j$.
	\\Otherwise, $w_{ij} = 0$, and $v_j = d(i, j)$.
	
	\textbf{Case 2.} Some closed facility $i \notin S$ is the bottleneck for $j$, and $j$ does not contribute to any open facility (i.e. $s_j = 0$).
	\\There must be some open facility $k \in S$ that caused $i$ to close. Connect $j$ to $k$. By Lemma \ref{lem:duallemma4}, we know that $3(1+\epsilon)^7 v_j \ge d(k, j)$.
	
	\textbf{Case 3.} Some closed facility $i \notin S$ is the bottleneck for $j$, and there exists an open facility $\ell \in S$ with $w_{\ell j} > 0$, but $\ell$ was not the reason why $i$ was closed.
	\\Since $w_{\ell j} > 0$, $s_j = w_{\ell j}$, by the uniqueness of $s_j$. Connect $j$ to $\ell$.
	\\By Lemma \ref{lem:duallemma1}, we have that $\tilde{r}_{\ell} \ge \frac{1}{(1+\epsilon)^3} (d(\ell, j) + w_{\ell j})$. Also, there must be some open facility $k \in S$ which prevented $i$ from opening. Using similar reasoning as in the previous case, we have that $d(k, j) \le 3(1+\epsilon)^7 v_j$. Now, 
	\begin{align*}
	&d(\ell, k) \ge 2(1+\epsilon)^3 r_{\ell} \ge 2 (d(\ell, j) + \beta w_{\ell j}) \tag{Using Claim \ref{cl:cl9.2} and Lemma \ref{lem:duallemma1}.}
	\\&\implies 2 (d(\ell, j) + \beta w_{\ell j}) \le d(\ell, k) \le d(\ell, j) + d(k, j) \tag{Triangle inequality}
	\\&\implies  d(\ell, j) + 2 \beta w_{\ell j} \le d(k, j) \le 3(1+\epsilon)^7 v_j
	\end{align*}
	
	\textbf{Case 4.} Some closed facility $i \notin S$ is the bottleneck for $j$. Furthermore, there is an open facility $k \in S$ such that $w_{kj} > 0$, and $k$ caused $i$ to be closed. Connect $j$ to $k$.
	\\Again, by uniqueness of $s_j$, we have that $s_j = w_{kj}$. Also, from Lemma \ref{lem:duallemma2}, $3(1+\epsilon)^3 v_j \ge 2\tilde{r}_i$. Since $k$ caused $i$ to be closed, we have that $\tilde{r}_i \ge \tilde{r}_k \ge \frac{1}{(1+\epsilon)^3} (d(k, j) + \beta w_{kj}) = \frac{1}{(1+\epsilon)^3} (d(k, j) + \beta s_j)$, by Lemma \ref{lem:duallemma2}. Combining the previous inequalities yields, $3(1+\epsilon)^6 v_j \ge d(k, j) + 2\beta s_j$.
	
	Finally, we use the well-known fact that for any $\epsilon \in (0, 1)$, $(1+\epsilon)^7 \le (1+c\epsilon)$ for some constant $c$, and the lemma follows.
\end{proof}  
\end{document}